\documentclass[letterpaper, 10 pt, conference]{ieeeconf}  % Comment this line out if you need a4paper

\IEEEoverridecommandlockouts                              % This command is only needed if 
\usepackage[dvipdfmx]{graphicx} % for pdf, bitmapped graphics files
\usepackage{amsmath} % assumes amsmath package installed
\usepackage{amssymb}  % assumes amsmath package installed
\usepackage{amsfonts}  % assumes amsmath package installed

\newtheorem{thm}{Theorem}
\newtheorem{cor}[thm]{Corollary}

\newtheorem{prop}[thm]{Proposition}

\newtheorem{defn}[thm]{Definition}

\title{\LARGE \bf
A smoothing theory for open quantum systems: \\ The least mean square approach
}

\author{Kentaro Ohki$^{1}$% <-this % stops a space
\thanks{*This work was supported by JSPS KAKENHI Grant Number 80241941 and 25820178}% <-this % stops a space
\thanks{$^{1}$Kentaro Ohki is with Graduate School of Informatics, Kyoto University, Yoshida-Honmachi, Sakyo-ku, Kyoto, Japan
        {\tt\small ohki@i.kyoto-u.ac.jp}}%
}

\begin{document}

\maketitle
\thispagestyle{empty}
\pagestyle{empty}

%%%%%%%%%%%%%%%%%%%%%%%%%%%%%%%%%%%%%%%%%%%%%%%%%%%%%%%%%%%%%%%%%%%%%%%%%%%%%%%%
\begin{abstract}
Unlike the classical smoothing theory, 
it is well known that quantum smoothers are, in general, not well--defined by the quantum conditional expectation. 
The purpose of this paper is to propose a new quantum smoothing theory 
based on the least mean squared errors.  
The least mean square estimate of quantum physical quantity composes from symmetric part and skew part, 
and we developed the recursive equations, respectively.  
\end{abstract}

%%%%%%%%%%%%%%%%%%%%%%%%%%%%%%%%%%%%%%%%%%%%%%%%%%%%%%%%%%%%%%%%%%%%%%%%%%%%%%%%
\section{INTRODUCTION}

	Estimation theory in quantum physics has now become an important field in technologies 
	and it has been studied in recent decades \cite{paris2004quantum,petz2008qit,busch2014quantum}.  
	One of the most important estimation methods is the least mean square estimation and 
	it is widely used in classical and quantum statistics \cite{busch2014quantum}.  
	When an estimand obeys a dynamical equation, 
	we sometimes develop a dynamical estimator which is driven by measurement outcomes.  
	The dynamical least mean square estimator is called a {\it filter} if the time of the estimand is the end of the interval of the measurement records and is called a {\it smoother} if the time of the estimand is in the interval of the measurement records  \cite{krishnan1984nonlinear}.  
	A key notion for the classical least mean square estimation is the conditional expectation and 
	it was expanded to quantum statistics.   
	The quantum filtering theory based on the quantum conditional expectation was initiated by Belavkin \cite{Belavkin_nondemolition,Bouten_quant_filtering,Bouten_discrete_quant_filtering} and it has been applied to many of quantum control problems \cite{wiseman_milburn_q-control,altafini2012modeling,belavkin2007emq,dong2009survey,james2014quantum,zhang2014quantum}.  
	
	On the other hand, quantum smoothers are not always described by the quantum conditional expectation 
	because of the non-commutative nature among quantum physical quantities.   
	Some recent studies provide new smoothing methods  \cite{yanagisawa2007smoothing,tsang2009time,tsang2014bayesian,guevara2015quantum,gammelmark2013past} .  
	Especially, Gammelmark et al. \cite{gammelmark2013past} developed a fixed interval smoothing theory based on the {\it weak value} \cite{aharonov1988result,aharonov2008quantum}, which is recently focused in quantum theory  \cite{tsang2014bayesian,guevara2015quantum,gammelmark2013past,das2014estimation,combes2014quantum,li2014amplification,dressel2014understanding} and experiments \cite{campagneibarcq2014observing,rybarczyk2014past,tan2015prediction}.  
	This Bayesian approach gives complex valued estimate and therefore its physical interpretation is often discussed \cite{dressel2014understanding}.  
	The weak value is expected as a new quantum metrology for such objectives as detecting quantum wave forms \cite{tsang2009time} and single photon detection \cite{li2014amplification}, and this is one of the goals of developing quantum smoothers.  
%	to measure precisely the past quantum observables, e.g., detection of arriving the particles like graviton \cite{dressel2014understanding}.    

	In this paper, we revisit the least mean square method for quantum dynamical systems to develop quantum smoothers on classical probabilistic space.  
	In contrast to Gammelmark et al.'s work, our proposal smoother is implemented by forward equations.  
	We introduce two pre--inner products and show two least mean square estimates which depend on each pre--inner product.  
	Another least mean square estimation approach for quantum dynamical systems was developed by Amini et al. \cite{amini2014quantum}.  
	They derived 	the linear least mean square estimator for linear quantum systems 
	and their estimator is realized in quantum systems.  
	It implies the estimates should be physical quantities and therefore cannot give non-causal estimates.  
	In contrast to their estimator, ours are computed in classical computers, i.e., 
	it is not necessary that each estimate is realized as a quantum physical quantity; 
	in fact, our estimator can give complex-valued estimates.

	There are two contributions in this paper.  
	First, we introduced the symmetric pre--inner product and 
	characterized the symmetric part and skew part of the quantum least mean square estimate.  
	We show there are two optimal estimations which depend on pre--inner products.  
	Second, we developed new quantum smoothers for fixed point smoothing problems.  
	Though we only show the quantum smoothers for fixed point smoothing problems, 
	the derivation is easily expanded to other smoothing problems, e.g. fixed lag smoothing problems.  

	This paper is organized as follows; 
	we review basic mathematical formulation of quantum theory in Section 2.  
	In Section 3, two inner products and their properties are shown.  
	We then introduce a dynamics of an open quantum system to be estimated and its filtering equation in Section 4.  
	Finally, the proposal quantum smoother is derived in Section 5.

	\subsection*{Notation}
	$\mathbb{R}$ and $\mathbb{C}$ are real numbers and complex numbers, respectively,   
	and $\mathrm{i}:= \sqrt{-1}$.  
	${\mathcal{H}}$ is a complex Hilbert space and we also denote ${\mathcal{H}}_{X}$ if it is the Hilbert space of the system $X$.  
	Any linear operator on a Hilbert space $\mathcal{H}$ is denoted by hat, e.g., $\hat{X}$.  
	When positive operators $\hat{X}$ and $\hat{Y}$ satisfy $\hat{X} = \hat{Y}^{2}$, we denote $\hat{Y} = \sqrt{\hat{X}}$.  
	The absolute value of operator is defined $|\hat{X}| := \sqrt{\hat{X}^{\ast} \hat{X}}$.  
	$\mathcal{L}(\mathcal{H}) $ is a set of linear bounded operators on the Hilbert space $\mathcal{H}$.  
	$\hat{X}\geq 0$ means that $\hat{X} \in \mathcal{L}(\mathcal{H})$ is a positive operator and $\hat{X}^{\ast}$ implies the conjugate operator of $\hat{X}$.  
	$\mathrm{Tr}[\bullet ] : \mathcal{L}(\mathcal{H}) \to \mathbb{C}$ is the trace on linear operators.  
	${\mathcal{S}}({\mathcal{H}}) := \{ \hat{\rho} \in {\mathcal{L}}({\mathcal{H}}) \ | \ 
	\hat{\rho} \geq 0,\ {\mathrm{Tr}}[\hat{\rho}] = 1 \}$ is a set of density operators.  
	$\hat{1}_{\mathcal{H}}$ is the identity operator on $\mathcal{H}$ and we sometimes omit its subscript.  
	Denote $[\hat{X} , \hat{Y} ]_{\pm} := \hat{X}\hat{Y} \pm \hat{Y} \hat{X}$, $\forall \hat{X},\hat{Y} \in \mathcal{L}(\mathcal{H})$.  
	$\otimes$ represents the Kronecker product for matrices and the tensor product for operators or spaces.

\section{BASIC DEFINITIONS OF QUANTUM THEORY}

	Here we briefly review the quantum theory.  
	For details, see, e.g., \cite{wiseman_milburn_q-control,holevo1982pas}.  
	
%	\subsection{Quantum theory}
	
	Every quantum system is described by a suitably defined Hilbert space $\mathcal{H}$.  
	All of the quantum physical quantities are denoted by self-adjoint operators on $\mathcal{H}$.  
	In this paper, we only consider linear bounded operators except quantum noise operators.  
	We denote a set of linear bounded operators on $\mathcal{H}$ by $\mathcal{L}( \mathcal{H})$.  
	The observation of any quantum physical quantity is a randomly chosen number 
	from the spectrum of the corresponding self-adjoint operator.  
	Random outcomes of all bounded operators make the quantum statistics 
	and the quantum expectation $\mathbb{P}_{\hat{\rho}}$ is defined as 
	${\mathbb{P}}_{\hat{\rho}} [\hat{X} ] := {\mathrm{Tr}}[ \hat{\rho} \hat{X}]$, 
	$\hat{\rho} \in \mathcal{S}(\mathcal{H}) $.  
%	:= \{ \hat{\rho} \in \mathcal{L}(\mathcal{H}) \ | \ \hat{\rho} \geq 0 , \ \mathrm{Tr}[\hat{\rho}]  =1\}$.  
	The quantum version of $\sigma$--measurable functions is von Neumann algebra, which is, roughly speaking, an algebra generated by projection operators with algebraic operations \cite{Bouten_discrete_quant_filtering}.  
	Let $\mathcal{A} \subseteq \mathcal{L}(\mathcal{H})$ be a von Neumann subalgebra.  
	A pair $(\mathcal{A} , \mathbb{P}_{\hat{\rho}} )$ is called {\it the quantum probability space}.  
	For a given quantum probability space $(\mathcal{A} , \mathbb{P}_{\hat{\rho}} )$, 
	a subalgebra $\mathcal{N}_{\hat{\rho}} := 
	\{ \hat{X} \in \mathcal{A} \ | \ \mathbb{P}_{\hat{\rho}} [\hat{X}^{\ast}\hat{X}] = 0  \} $ 
	of $\mathcal{A}$ is a quantum version of the measure zero set with respect to $\mathbb{P}_{\hat{\rho}}$, 
	called the left kernel of $\mathbb{P}_{\hat{\rho}}$  \cite{takesaki2003toa1}.  
	The left kernel $\mathcal{N}_{\hat{\rho}}$ is not empty because it always includes $0$.   
	Moreover, $\mathcal{N}_{\hat{\rho}}$ is a left ideal and satisfies 
	$
%	\begin{align*}
	\mathbb{P}_{\hat{\rho} } \left[ ( \hat{X} + \hat{Z}_{1} )^{\ast} ( \hat{Y} + \hat{Z}_{2} ) \right]
	=
	\mathbb{P}_{\hat{\rho} } \left[ \hat{X}^{\ast} \hat{Y} \right]
%	\end{align*}
	$
	for any $\hat{X}, \hat{Y} \in \mathcal{A}$ and $\hat{Z}_{1} , \hat{Z}_{2} \in \mathcal{N}_{\hat{\rho}}$ 
	(see, e.g., \cite[Lemma 9.6 in Chap. 1]{takesaki2003toa1}).  
	If for $\hat{X} ,\hat{Y} \in \mathcal{A}$, there exists $\hat{Z} \in \mathcal{N}_{\hat{\rho}}$ s.t. $\hat{X} = \hat{Y} + \hat{Z}$, then we denote $\hat{X} = \hat{Y}$, $\mathbb{P}_{\hat{\rho}}$--a.s. or $\hat{\rho}$--a.s. for short.

\section{LEAST MEAN SQUARE ESTIMATION}

	In this section, we introduce the symmetric and skew part of the best estimation 
	in the sense of the semi-norm defined by the pre-inner product below and show several properties of them.  
	Let $\mathcal{Y}$ be a commutative $\ast$--subalgebra of $\mathcal{L}(\mathcal{H})$.  
	We introduce another $\ast$--subalgebra whose elements commute with all of the elements in $\mathcal{Y}$; 
	\begin{align*}
	\mathcal{Y}^{\prime} := 
	\{  \hat{X} \in \mathcal{L}(\mathcal{H}) 
	\ | \  
	\hat{X} \hat{Y} = \hat{Y}\hat{X} ,\ \forall \hat{Y} \in \mathcal{Y}  
	\}
	.
	\end{align*}
	$\mathcal{Y}^{\prime}$ is called a commutant of $\mathcal{Y}$ in $\mathcal{L}(\mathcal{H})$.  
	Hereafter we assume $\mathcal{Y} = (\mathcal{Y} ^{\prime}) ^{\prime}$, i.e., $\mathcal{Y}$ is a commutative von Neumann subalgebra \cite{Bouten_quant_filtering,takesaki2003toa1}.  
%	For instance, $\mathcal{Y} = \mathrm{alg}( \{ \hat{P}_{j}\} _{j=1}^{m} )$ is a commutative von Neumann subalgebra in $\mathbb{C}^{n\times n}$.  
	Every von Neumann algebra is a generalization of the set of the $\sigma$--measurable bounded functions and especially 
	every commutative von Neumann algebra is isomorphic to the corresponding set of the $\sigma$--measurable bounded functions.  
	Note that $\mathcal{Y}^{\prime}$ is generally non-commutative $\ast$--subalgebra.

	\subsection{Definitions} 
	
	We introduce three approximations of a given $\hat{X} \in \mathcal{L} ( \mathcal{H} )$ here 
	and show that two of them provides the best estimations in different measures in following subsection.  
	All of them is based on the following pre-inner products \cite{amari2007mig}.  
	
	\begin{defn}
	\label{defn_cdc2015_inner_products}
	
	For given $\hat{\rho} \in \mathcal{S}(\mathcal{H})$, 	
	\begin{enumerate}
	\item {\rm the pre--inner product} $\langle \bullet , \bullet \rangle _{\hat{\rho}} 
	:
	\mathcal{L}( \mathcal{H} ) \times \mathcal{L}( \mathcal{H} ) \to \mathbb{C}
	$ 
	is defined by $\langle \hat{X} , \hat{Y} \rangle _{\hat{\rho}} := \mathbb{P} _{\hat{\rho}} \left[ \hat{X} ^{\ast} \hat{Y} \right]$.  
	
	\item {\rm the symmetric pre--inner product} $\langle \langle \bullet , \bullet \rangle \rangle _{\hat{\rho}} 
	:
	\mathcal{L}( \mathcal{H} ) \times \mathcal{L}( \mathcal{H} ) \to \mathbb{C}
	$ is 
	defined by $\langle \langle \hat{X} , \hat{Y} \rangle \rangle _{\hat{\rho}} := 
	\frac{1}{2} \mathbb{P} _{\hat{\rho}} \left[ \hat{X} ^{\ast} \hat{Y} + \hat{Y} \hat{X} ^{\ast} \right]$.  
	
	\end{enumerate}
	
	\end{defn}
	
	Note that $\langle \langle \hat{X} , \hat{Y} \rangle \rangle _{\hat{\rho}}$  is not 
	the real part of $\langle \hat{X} , \hat{Y} \rangle _{\hat{\rho}}$.  
	These pre--inner products satisfy the Cauchy--Schwarz inequality (see, e.g., \cite[Prop. 9.5 in Chap. 1]{takesaki2003toa1}).   
	$\langle \hat{X} ,\hat{X} \rangle _{\hat{\rho}} =0$ is the necessary and sufficient condition for 
	$\hat{X} \in \mathcal{N}_{\hat{\rho}}$, though, 
	$\hat{X} \in \mathcal{N}_{\hat{\rho}}$ does not implies $\langle \langle \hat{X} ,\hat{X} \rangle \rangle _{\hat{\rho}} =0$.  
	This does not bother our exploration of the least means square estimation because 
%	This confliction bothers our  exploration 
	if $\hat{X} \in \mathcal{N}_{\hat{\rho}} \cap \mathcal{Y} $, then $\langle \hat{X} ,\hat{X} \rangle _{\hat{\rho}} = \langle \langle \hat{X} ,\hat{X} \rangle \rangle _{\hat{\rho}} =0$.  
	This is shown by the Cauchy--Schwarz inequality and commutativity of $\mathcal{Y}$.  
	We use two measures to find the best approximation in $\mathcal{Y}$, and the null space there is common between them.

%	For $\hat{X} \geq 0$ and $\hat{Y}\geq 0$, $\langle \hat{X} , \hat{Y} \rangle _{\hat{\rho}} \geq 0$ does not hold generally.  
	Let us define the quantum conditional expectation (see, e.g., \cite[Sec. 3]{Bouten_quant_filtering} or \cite[Prop. 2.36]{takesaki2003toa1}).

	\begin{defn}[Quantum conditional expectation]
	
	Let $(\mathcal{L}(\mathcal{H}) , \mathbb{P}_{\hat{\rho}})$ be a quantum probability space and $\mathcal{Y}$ be a commutative von Neumann sub--algebra of $\mathcal{L}(\mathcal{H})$.  
	A linear operator $\hat{Q} \in \mathcal{Y}$ is called a version of the {\it quantum conditional expectation} 
	if there exists $\hat{Q} \in \mathcal{Y}$ satisfies
	\begin{align}
	\langle \hat{Z} , \hat{X} - \hat{Q} \rangle _{\hat{\rho}} = 0
	,\quad \forall \hat{Z} \in \mathcal{Y}
	\label{eq_quantum_conditional_expectation}
	\end{align}
	for arbitrary fixed $\hat{X} \in \mathcal{Y}^{\prime}$.  
	Then we denote $\hat{Q} = \mathbb{P}_{\hat{\rho}} \left[ \hat{X} | \mathcal{Y} \right] $.  
	
	\end{defn}
	
	Some properties of the quantum conditional expectation are shown in, for example, \cite{Bouten_quant_filtering}.  
	The definition of the quantum conditional expectation implies that the $\hat{X} - \mathbb{P}_{\hat{\rho}} \left[ \hat{X} | \mathcal{Y} \right]$ and the commutative sub-algebra $\mathcal{Y} $ are orthogonal under state $\mathbb{P}_{\hat{\rho}}$.  
	We extend the definition of orthogonality to non-commutative regime.

	\begin{defn}
	\label{def_least_square_approximation}
	\label{def_anti_synmetric_approximation}
	
	Let $(\mathcal{L}(\mathcal{H}) , \mathbb{P}_{\hat{\rho}})$ be a quantum probability space 
	and $\mathcal{Y}$ be a commutative von Neumann sub--algebra of $\mathcal{L}(\mathcal{H})$.  
	For arbitrary fixed $\hat{X} \in \mathcal{L}(\mathcal{H})$, we define following operators: 
	\begin{enumerate}
	\item 
	A linear operator $\hat{Q} \in \mathcal{Y}$ is called a version of {\it symmetric quantum least mean square estimate} 
	if there exists $\hat{Q} \in \mathcal{Y}$ that satisfies
	\begin{align}
	\langle \langle \hat{Z} , \hat{X} - \hat{Q}\rangle \rangle _{\hat{\rho}} = 0
	%\mathbb{P}_{\hat{\rho}} \left( \hat{X} \hat{Z} + \hat{Z}\hat{X} \right) 
	%= 
	%2\mathbb{P}_{\hat{\rho}} \left( \hat{Y} \hat{Z} \right) 
	,\quad \forall \hat{Z} \in \mathcal{Y}
	.
	\label{eq_least_square_approximate}
	\end{align}  
	Then we denote $\hat{Q} = \mathbb{Q}_{\hat{\rho}}^{+} \left[ \hat{X} | \mathcal{Y} \right] $. 
	
	\item 
	A linear operator $\hat{Q} \in \mathcal{Y}$ is called a version of {\it the mean non-commutativity with respect to $\mathcal{Y}$} 
	if there exists $\hat{Q} \in \mathcal{Y}$ that satisfies
	\begin{align}
	\mathbb{P}_{\hat{\rho}} \left[ \hat{Z} \hat{X}  - \hat{X} \hat{Z} \right] 
	= 
	2\mathbb{P}_{\hat{\rho}} \left[ \hat{Q} \hat{Z} \right] 
	, \quad \forall \hat{Z} \in \mathcal{Y}
	.
	\label{eq_commutation_approximation}
	\end{align}
	Then we denote $\hat{Q} = \mathbb{Q}_{\hat{\rho}} ^{-} \left[ \hat{X} | \mathcal{Y} \right] $.

	\item $\mathbb{Q}_{\hat{\rho}}  \left[ \hat{X} | \mathcal{Y} \right] 
	:= 
	\mathbb{Q}_{\hat{\rho}} ^{+} \left[ \hat{X} | \mathcal{Y} \right] 
	+ 
	\mathbb{Q}_{\hat{\rho}} ^{-} \left[ \hat{X} | \mathcal{Y} \right]$ is called {\it the quantum least mean square estimate of $\hat{X}$ with respect to $\mathcal{Y}$}.

	\end{enumerate}

	\end{defn}

	Roughly speaking, $ \mathbb{Q}_{\hat{\rho}}^{+} \left[ \hat{X} | \mathcal{Y} \right]$ and 
	$ \mathbb{Q}_{\hat{\rho}}^{-} \left[ \hat{X} | \mathcal{Y} \right]$ are the symmetric part and the skew part of the least mean square estimate, respectively.  
	In fact, it is true for $\hat{X} = \hat{X}^{\ast}$.  
	We call Eq. \eqref{eq_least_square_approximate} and \eqref{eq_commutation_approximation} {\it the symmetric orthogonal condition} and {\it the skew symmetric orthogonal condition}, respectively.  
	If $\mathcal{N} _{\hat{\rho}} \cap \mathcal{Y} \neq \{ 0 \} $, then there are many operators that satisfy above conditions.  
	This is why we use ``a version of."  
%		This is not a quantum conditional expectation in the sense of Umegaki's quantum conditional expectation \cite{takesaki2003toa1}.  
	Obviously, $\mathbb{P}_{\hat{\rho}} \left[ \hat{X} \right] 
	= 
	\mathbb{P}_{\hat{\rho}} \left[ \mathbb{Q}_{\hat{\rho}} \left[ \hat{X} | \mathcal{Y} \right]  \right]  
	=
	\mathbb{P}_{\hat{\rho}} \left[ \mathbb{Q}_{\hat{\rho}}^{+} \left[ \hat{X} | \mathcal{Y} \right]  \right]  
	$ holds, i.e., these are unbiased estimates.  
%	$\mathbb{P}_{\hat{\rho}} \left[ \mathbb{Q}_{\hat{\rho}} ^{-} \left[ \hat{X} | \mathcal{Y} \right] \right] =0$ also holds.  
	The name ``the symmetric quantum least mean square estimate" 
	and 
	``the quantum least mean square estimate" are originated from 
	Proposition \ref{prop_least_mean_squareness} below.  
	Since these estimates do not satisfy the positivity in general, neither of them is a version of quantum conditional expectation.   
	Note that from the linearity of the quantum expectation, 
	\begin{align*}
	\mathbb{P}_{\hat{\rho}} \left[ \hat{Z} \hat{X}  \right]
	=
	\mathbb{P}_{\hat{\rho}} \left[ 
	\mathbb{Q}_{\hat{\rho}}^{+} \left[ \hat{X} | \mathcal{Y} \right]  \hat{Z} 
	\right] 
	+ 
	\mathbb{P}_{\hat{\rho}} \left[ 
	\mathbb{Q}_{\hat{\rho}}^{-} \left[ \hat{X} | \mathcal{Y} \right] \hat{Z} 
	\right] 
	\end{align*}
	implies the orthogonality in the sense of $\langle \bullet , \bullet \rangle _{\hat{\rho}}$; 
	\begin{align*}
	\langle \hat{Z} , \hat{X} - \mathbb{Q}_{\hat{\rho}} \left[ \hat{X} | \mathcal{Y} \right] \rangle _{\hat{\rho}}
	=
	0
	, \quad 
	\forall \hat{Z} \in \mathcal{Y}
	.  
	\end{align*}
	
%	Since the expectation of $\mathbb{P}_{\hat{\rho}} \left[ \mathbb{Q}_{\hat{\rho}} ^{-} \left[ \hat{X} | \mathcal{Y} \right] \right] $ is always zero, 
	Statistical and physical meaning of $\mathbb{Q}_{\hat{\rho}} ^{-} \left[ \hat{X} | \mathcal{Y} \right] $ is unclear, though, 
	it plays a role as a measure of noncommutativity between $\hat{X}$ and $\mathcal{Y}$.  
%	However, this is an interesting quantity in the view of non-commutative geometry.   
	If $\hat{X}$ and $\hat{Z}$ are Hilbert--Schmidt class operators, 
	$[\hat{X} , \hat{Z} ]_{-}$ is orthogonal to both of $\hat{X}$ and $\hat{Z}$ in the sense of Hilbert--Schmidt inner product.  
	The operator $\mathbb{Q}_{\hat{\rho}}^{-} \left[ \hat{X} | \mathcal{Y} \right]$ represents 
	a measure of the ``$\hat{\rho}$--direction" component of the orthogonal direction 
	against to the both of $\hat{X}$ and $\mathcal{Y}$.

	\subsection{Basic properties}
	
%	The properties of these operators are shown below.  
%	First, we prove the linearity and uniqueness of $\mathbb{Q}_{\hat{\rho}}^{\pm} \left[ \hat{X} \ | \ \mathcal{Y} \right]$ , 
%	and self-adjointness or skew self-adjointness.  
	
	From the definitions, following properties hold.  
	\begin{enumerate}
	
%	\item {\bf (existence)} 
	
	\item {\bf (linearity)} $\mathbb{Q}_{\hat{\rho}}^{\pm} \left[ \bullet \ | \ \mathcal{Y} \right] : \mathcal{L}(\mathcal{H}) \to \mathcal{Y}$ is linear.  
		
	\item  {\bf (self-adjointness and skewness)} $\mathbb{Q}_{\hat{\rho}}^{\pm} \left[ \hat{X} \ | \ \mathcal{Y} \right] ^{\ast} = \pm \mathbb{Q}_{\hat{\rho}}^{\pm} \left[ \hat{X} \ | \ \mathcal{Y} \right]$, 
	$\mathbb{P}_{\hat{\rho}}$--a.s.  
	for $\hat{X} = \hat{X}^{\ast} \in \mathcal{L}(\mathcal{H})$, 
	and 
	$\mathbb{Q}_{\hat{\rho}}^{\pm} \left[ \hat{X} \ | \ \mathcal{Y} \right] ^{\ast} = \mp \mathbb{Q}_{\hat{\rho}}^{\pm} \left[ \hat{X} \ | \ \mathcal{Y} \right]$, $\mathbb{P}_{\hat{\rho}}$--a.s. for 
	$\hat{X} = -\hat{X}^{\ast} \in \mathcal{L}(\mathcal{H})$.  
	
	\end{enumerate}
	
	From above properties, the $\mathbb{Q}_{\hat{\rho}} \left[ \hat{X} \ | \ \mathcal{Y} \right]$ could be a normal operator even if $\hat{X}$ is a self-adjoint operator.  
	In this sense, $\mathbb{Q}_{\hat{\rho}}^{\pm} \left[ \hat{X} \ | \ \mathcal{Y} \right]$ are real and imaginary part of $\mathbb{Q}_{\hat{\rho}} \left[ \hat{X} \ | \ \mathcal{Y} \right]$, respectively.  
	Moreover, the uniqueness also holds.  
	%$\mathbb{Q}_{\hat{\rho}}^{\pm} \left[ \hat{X} \ | \ \mathcal{Y} \right]$ is uniquely determined in the sense of $\mathbb{P}_{\hat{\rho}}$--a.s.  

	\begin{prop}[Uniqueness]
	\label{prop_uniqueness}
	
	For any $\hat{X} \in \mathcal{L}(\mathcal{H})$, $\mathbb{Q}_{\hat{\rho}}^{\pm} \left[ \hat{X} \ | \ \mathcal{Y} \right]$ is uniquely determined $\mathbb{P}_{\hat{\rho}}$--a.s.  
	
	\end{prop}
	
	\begin{proof}
	Suppose $\hat{Y}_{1},\ \hat{Y}_{2} \in \mathcal{Y}$ satisfy the symmetric orthogonal condition \eqref{eq_least_square_approximate} for $\hat{X}$.  
	Then, 
	$\langle \langle \hat{X} - \hat{Y}_{2} , \hat{Z} \rangle \rangle _{\hat{\rho}} = 0$ minus 
	$\langle \langle \hat{X} - \hat{Y}_{1} , \hat{Z} \rangle \rangle _{\hat{\rho}} = 0$ gives 
	\begin{align*}
	\langle \langle \hat{Y}_{1} - \hat{Y}_{2} , \hat{Z} \rangle \rangle _{\hat{\rho}}
	=
	\langle \hat{Y}_{1} - \hat{Y}_{2} , \hat{Z} \rangle _{\hat{\rho}}
	=0, 
	\quad
	\forall \hat{Z} \in \mathcal{Y}
	.
	\end{align*} 
	Since $\hat{Y}_{1} - \hat{Y}_{2} \in \mathcal{Y}$, choosing $\hat{Z} = ( \hat{Y}_{1} - \hat{Y}_{2} ) $ makes 
	$\mathbb{P}_{\hat{\rho}} \left[ (\hat{Y}_{1} - \hat{Y}_{2} )^{\ast} (\hat{Y}_{1} - \hat{Y}_{2} ) \right] =0$.  
	It implies $\hat{Y}_{1} = \hat{Y}_{2}$ under the state $\mathbb{P}_{\hat{\rho}}$ or $\hat{Y}_{1} - \hat{Y}_{2} \in \mathcal{N}_{\hat{\rho}}$.  
	
	Similarly, suppose $\hat{Y}_{3},\ \hat{Y}_{4} \in \mathcal{Y}$ satisfy the skew symmetric orthogonal condition \eqref{eq_commutation_approximation} for $\hat{X}$.  
	Then, 
	\begin{align*}
	\langle \hat{Y}_{3} - \hat{Y}_{4} , \hat{Z} \rangle _{\hat{\rho}}
	=0, 
	\quad
	\forall \hat{Z} \in \mathcal{Y}, 
	\end{align*} 
	and choosing $\hat{Z} = (\hat{Y}_{3} - \hat{Y}_{4} )$ shows the uniqueness holds $\mathbb{P}_{\hat{\rho}}$--a.s.  

	\end{proof}
	
	Since $\langle \hat{X} , \hat{X} \rangle _{\hat{\rho}} 
	\neq \langle \langle \hat{X} , \hat{X} \rangle \rangle _{\hat{\rho}}$ in general, 
	we can consider two different least mean square estimations.  
	$\mathbb{Q}_{\hat{\rho}} \left[ \hat{X} | \mathcal{Y} \right]$ and 
	$\mathbb{Q}_{\hat{\rho}}^{+} \left[ \hat{X} | \mathcal{Y} \right]$ are 
	the least mean error estimates in the sense of each semi-norm defined by each pre-inner product in Definition \ref{defn_cdc2015_inner_products}, respectively.  
	
	% Least mean square estimation
	%
	\begin{prop}
	\label{prop_least_mean_squareness}
	\label{prop_least_mean_squareness_symmetric}\
	
	\begin{enumerate} 
	\item 
	For arbitrary $\hat{X} \in \mathcal{L}(\mathcal{H})$, 
	\begin{align*}
	& \langle 
	\hat{X} - \mathbb{Q}_{\hat{\rho}} \left[ \hat{X} | \mathcal{Y} \right] , 
	\hat{X} - \mathbb{Q}_{\hat{\rho}} \left[ \hat{X} | \mathcal{Y} \right] 
	\rangle _{\hat{\rho}}
	\nonumber \\ 
	\leq &
	\langle 
	\hat{X} - \hat{Z} , 
	\hat{X} - \hat{Z} 
	\rangle _{\hat{\rho}}, 
	\quad  \forall \hat{Z} \in \mathcal{Y}
	\nonumber .
	\end{align*}

	\item For arbitrary $\hat{X} \in \mathcal{L}(\mathcal{H})$, 
	\begin{align*}
	& \langle \langle 
	\hat{X} - \mathbb{Q}_{\hat{\rho}}^{+}  \left[ \hat{X} | \mathcal{Y} \right]  , 
	\hat{X} - \mathbb{Q}_{\hat{\rho}}^{+}  \left[ \hat{X} | \mathcal{Y} \right] 
	\rangle \rangle _{\hat{\rho}}
	\nonumber \\ \leq &
	\langle \langle 
	\hat{X} - \hat{Z} , 
	\hat{X} - \hat{Z} 
	\rangle \rangle _{\hat{\rho}}
	,\quad  \forall \hat{Z}  \in \mathcal{Y}
	\nonumber .
	\end{align*}

	\end{enumerate}

	\end{prop}
	
%	Proof is in Appendix \ref{proof_prop_least_mean_squareness}.  
\begin{proof}

	\begin{enumerate}
	
	\item 
	Let $\hat{Y} = \mathbb{Q}_{\hat{\rho}}  \left[ \hat{X} | \mathcal{Y} \right]$ 
	and 
	$\hat{Y} ^{\pm}= \mathbb{Q}_{\hat{\rho} } ^{\pm}  \left[ \hat{X} | \mathcal{Y} \right] $.  
	For every $\hat{Z} \in \mathcal{L}(\mathcal{H})$, 
	\begin{align*}
	&
	\langle 
	\hat{X} - \hat{Z} , 
	\hat{X} - \hat{Z} 
	\rangle _{\hat{\rho}}
	\\
	= &
	\langle 
	\hat{X} - \hat{Y} , 
	\hat{X} - \hat{Y} 
	\rangle _{\hat{\rho}}
	+
	\langle 
	\hat{Y} - \hat{Z} , 
	\hat{Y} - \hat{Z} 
	\rangle _{\hat{\rho}}
	\\
	& +
	\langle 
	  \hat{X} - \hat{Y}  , \hat{Y} - \hat{Z}  
	\rangle _{\hat{\rho}}
	+	
	\langle 
	 \hat{Y} - \hat{Z}  , \hat{X} - \hat{Y}  
	\rangle _{\hat{\rho}}
	.  
	\end{align*}
	Now we use the operator $\hat{R} := \hat{Y} - \hat{Z}$ for simplicity.  
	Since $\hat{Y} - \hat{Z} \in \mathcal{Y}$ and $\hat{Z}$ is arbitrary chosen,  
	$\hat{R}$ represents any element in $\mathcal{Y}$.  
	Then,from  Definitions \ref{def_least_square_approximation},  
	\begin{align*}
	\langle \hat{X} - \hat{Y} ,\hat{R} \rangle _{\hat{\rho}}
	=&
	\mathbb{P}_{\hat{\rho}} \left[ \hat{X} ^{\ast} \hat{R} \right] 
	- 
	\mathbb{P}_{\hat{\rho}} \left[ \hat{Y}^{\ast} \hat{R} \right]
	\\
	=&
	\langle \langle \hat{X} ,\hat{R} \rangle \rangle _{\hat{\rho}}
	- \mathbb{P}_{\hat{\rho}} \left[ (\hat{Y}^{+}) ^{\ast}\hat{R} \right] 
	\\ &
	+
	\frac{1}{2} \mathbb{P}_{\hat{\rho}} \left[ [ \hat{X}  ^{\ast} , \hat{R} ]_{-} \right] 
	+ \mathbb{P}_{\hat{\rho}} \left[ \hat{Y}^{-} \hat{R} \right] 
	\\
	=& 0
	.
	\end{align*}
	Similarly, $
	\langle \hat{R} , \hat{X} - \hat{Y}  \rangle _{\hat{\rho}}
	= 0$.  
	Finally, we obtain 
	\begin{align*}
	& \langle 
	\hat{X} - \hat{Z} , 
	\hat{X} - \hat{Z} 
	\rangle _{\hat{\rho}}
	\nonumber \\ 
	=&
	\langle 
	\hat{X} - \hat{Y} , 
	\hat{X} - \hat{Y} 
	\rangle _{\hat{\rho}}
	+
	\langle 
	\hat{Y} - \hat{Z} , 
	\hat{Y} - \hat{Z} 
	\rangle _{\hat{\rho}}
	\\
	\geq &
	\langle 
	\hat{X} - \hat{Y} , 
	\hat{X} - \hat{Y} 
	\rangle _{\hat{\rho}}
	.
	\end{align*}
	
	\item By a similar argument, our claim is proved.   

	\end{enumerate}

\end{proof}

	Our interest is whether richer information provides better estimation or not under two measures.  
	Proposition \ref{prop_least_mean_squareness} gives the following results.

	\begin{cor}
	
	Let $\mathcal{Y}_{1}$ and $\mathcal{Y}_{2}$  be commutative von Neumann subalgebras of $\mathcal{L}(\mathcal{H})$ and assume $\mathcal{Y}_{1} \subseteq \mathcal{Y}_{2}$.  
	Then, for any $\hat{X} = \hat{X}^{\ast} \in \mathcal{L}(\mathcal{H})$, 
	\begin{enumerate} 
	\item 
		\begin{align*}
		& \langle  \hat{X} - \mathbb{Q}_{\hat{\rho}}^{+} \left[ \hat{X} | \mathcal{Y}_{2} \right] 
		,
		\hat{X} - \mathbb{Q}_{\hat{\rho}}^{+} \left[ \hat{X} | \mathcal{Y}_{2} \right] 
		\rangle _{\hat{\rho}}
		\nonumber \\
		\leq &
		\langle  \hat{X} - \mathbb{Q}_{\hat{\rho}}^{+} \left[ \hat{X} | \mathcal{Y}_{1} \right] 
		,
		 \hat{X} - \mathbb{Q}_{\hat{\rho}}^{+} \left[ \hat{X} | \mathcal{Y}_{1} \right] 
		\rangle _{\hat{\rho}}
		.
		\end{align*}
	
	\item 
		\begin{align*}
		& \langle  \hat{X} - \mathbb{Q}_{\hat{\rho}}^{-} \left[ \hat{X} | \mathcal{Y}_{2} \right] 
		,
		\hat{X} - \mathbb{Q}_{\hat{\rho}}^{-} \left[ \hat{X} | \mathcal{Y}_{2} \right] 
		\rangle _{\hat{\rho}}
		\nonumber \\
		\leq &
		\langle  \hat{X} - \mathbb{Q}_{\hat{\rho}}^{-} \left[ \hat{X} | \mathcal{Y}_{1} \right] 
		,
		 \hat{X} - \mathbb{Q}_{\hat{\rho}}^{-} \left[ \hat{X} | \mathcal{Y}_{1} \right] 
		\rangle _{\hat{\rho}}
		.
		\end{align*}

	\item 
	\begin{align*}
	& \langle  \hat{X} - \mathbb{Q}_{\hat{\rho}} \left[ \hat{X} | \mathcal{Y}_{2} \right] 
	,
	\hat{X} - \mathbb{Q}_{\hat{\rho}} \left[ \hat{X} | \mathcal{Y}_{2} \right] 
	\rangle _{\hat{\rho}}
	\nonumber \\
	\leq &
	\langle  \hat{X} - \mathbb{Q}_{\hat{\rho}} \left[ \hat{X} | \mathcal{Y}_{1} \right] 
	,
	 \hat{X} - \mathbb{Q}_{\hat{\rho}} \left[ \hat{X} | \mathcal{Y}_{1} \right] 
	\rangle _{\hat{\rho}}
	.
	\end{align*}
		
	\end{enumerate}
	
	\end{cor}

	\begin{proof}
	We prove only 2).  
	Let $\hat{Y}_{i} :=  \mathbb{Q}_{\hat{\rho}}^{-} \left[ \hat{X} | \mathcal{Y}_{i} \right]$, $i=3,4$.  
		As $\hat{Y}_{3} \in \mathcal{Y}_{2}$, straightforward calculation gives 
		\begin{align*}
		&\langle  \hat{X} - \hat{Y}_{3},  \hat{X} - \hat{Y}_{3}	\rangle _{\hat{\rho}}
		\nonumber \\
		=&
		\langle  \hat{X} - \hat{Y}_{4} + \hat{Y}_{4} - \hat{Y}_{3},  
		\hat{X} - \hat{Y}_{4} + \hat{Y}_{4} - \hat{Y}_{3}	\rangle _{\hat{\rho}} 
		\nonumber \\
		=&
		\langle  \hat{X} - \hat{Y}_{4},  \hat{X} - \hat{Y}_{4}	\rangle _{\hat{\rho}}
		+
		\langle  \hat{Y}_{4} - \hat{Y}_{3},  \hat{Y}_{4} - \hat{Y}_{3}	\rangle _{\hat{\rho}}
		\nonumber \\ & \quad
		+
		\langle   \hat{X} - \hat{Y}_{4}  ,   \hat{Y}_{4} - \hat{Y}_{3}  \rangle _{\hat{\rho}}
		+
		\langle   \hat{Y}_{4} - \hat{Y}_{3}  ,   \hat{X} - \hat{Y}_{4}  \rangle _{\hat{\rho}}
			.
		\end{align*}
		Since $\hat{Y}_{i}^{\ast} = - \hat{Y}_{i}$, $i=3,4$, 
		\begin{align*}
		& \langle   \hat{X} - \hat{Y}_{4}  ,   \hat{Y}_{4} - \hat{Y}_{3}  \rangle _{\hat{\rho}}
		+
		\langle   \hat{Y}_{4} - \hat{Y}_{3}  ,   \hat{X} - \hat{Y}_{4}  \rangle _{\hat{\rho}}
		\nonumber \\
		=&
		\mathbb{P}_{\hat{\rho}}\left[
		\hat{X} \left( \hat{Y}_{4} - \hat{Y}_{3} \right) - \left( \hat{Y}_{4} - \hat{Y}_{3} \right) \hat{X} 
		 \right] 
		 \nonumber \\ & \quad
		 + 
		2
		 \mathbb{P}_{\hat{\rho}} \left[
		 \hat{Y}_{4} \left( \hat{Y}_{4} - \hat{Y}_{3} \right)
		 \right]
		 =0.  
		\end{align*}
		Therefore, 
		\begin{align*}
		&\langle  \hat{X} - \hat{Y}_{3},  \hat{X} - \hat{Y}_{3}	\rangle _{\hat{\rho}}
		\nonumber \\
		=&
		\langle  \hat{X} - \hat{Y}_{4},  \hat{X} - \hat{Y}_{4}	\rangle _{\hat{\rho}}
		+
		\langle  \hat{Y}_{4} - \hat{Y}_{3},  \hat{Y}_{4} - \hat{Y}_{3}	\rangle _{\hat{\rho}}
		\nonumber \\
		\geq &
		\langle  \hat{X} - \hat{Y}_{4},  \hat{X} - \hat{Y}_{4}	\rangle _{\hat{\rho}}
		.  
		\end{align*}

	\end{proof}

\section{MODEL AND QUANTUM FILTERING}

%%%%
	\subsection{Model}
%%%%
	
	Any quantum system is described by suitable Hilbert space and linear operators on the Hilbert space.  
	We consider two quantum systems, system to be estimated and probe system.  
	We describe their Hilbert spaces $\mathcal{H}_{S}$ and $\mathcal{H}_{P}$, respectively.  
	$\mathcal{H}_{P}$ is a continuous Fock space \cite{guichardet1969tensor}; $\mathcal{H}_{P} = \otimes _{t \in [0, \infty) } \mathcal{H}_{P}(t)$, where $\mathcal{H}_{P} (t)$ is a Hilbert space at time $t \geq 0$.  
	The compound quantum system is the tensor product Hilbert space $\mathcal{H} = \mathcal{H}_{S} \otimes \mathcal{H}_{P}$  equipped with a density operator $	\hat{\rho} = \hat{\rho}_{S} \otimes \hat{\rho}_{P},\quad 
	\hat{\rho}_{S}\in {\mathcal{S}}({\mathcal{H}}_{S}),\ 
	\hat{\rho}_{P}\in {\mathcal{S}}({\mathcal{H}}_{P})$.  
	For simplicity, we assume $\hat{\rho}_{P}$ is a vacuum state \cite{Gardiner_quantum_noise,walls2008quantum}.  
	Physical quantities of the system are described by self-adjoint operators in $\mathcal{L}(\mathcal{H}_{S})$ and 
	physical quantities of the probe system are described by self-adjoint operators in $\mathcal{L}(\mathcal{H}_{P})$.  
	They act on the total quantum system with corresponding identity operator, though, we omit identity operator for simplicity; 
	$\hat{X} \otimes \hat{1}_{P} \equiv \hat{X}$ and $\hat{1}_{S} \otimes \hat{Y} \equiv \hat{Y}$ for $\hat{X} \in \mathcal{L}(\mathcal{H}_{S})$ and $\hat{Y} \in \mathcal{L}(\mathcal{H}_{P})$.

	\begin{figure}[!htbp]
	\begin{center}
	\includegraphics[width=8.3cm]{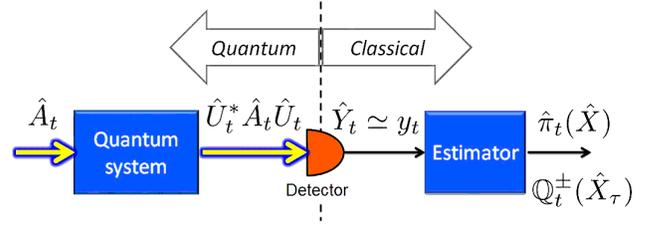}
	\caption{Schematic diagram}
	\label{fig_model_setup}
	\end{center}
	\end{figure}

	According to quantum theory, the time evolution of every physical quantity $\hat{X}=\hat{X}^{\ast} \in \mathcal{L}(\mathcal{H})$ driven by probe system is determined by a unitary operator $\hat{U}_{t}$ which denotes the interaction between the system and the probe.  
	We consider the unitary operator $\hat{U}_{t}$ as the solution of the following {\it Hudson--Parthasarathy equation}; 
	\begin{align}
	d\hat{U}_{t} =
	\left( 
	- 
	\mathrm{i} \hat{H}dt 
	- 
	\frac{1}{2}\hat{L}^{\ast}\hat{L} dt 
	+ 
	\hat{L}d\hat{A}_{t} ^{\ast} 
	- 
	\hat{L}^{\ast} 
	d\hat{A}_{t} 
	\right)
	\hat{U}_{t}
	\label{eq_hudson_parthasarathy_equation}
	\end{align}
	with $\hat{U}_{0} = \hat{1}$, 
	where $\hat{H}, \hat{L} \in \mathcal{L}(\mathcal{H}_{S})$ and 
	$\hat{A}_{t} \in \mathcal{L}(\mathcal{H}_{P})$ is a quantum anihilation process which satisfies quantum Ito's rule \cite{Hudson_Parthasarathy};  
	\begin{align}
	\left\{
	\begin{array}{l}
	d\hat{A}_{t} d\hat{A}_{t} = d\hat{A}_{t}^{\ast} d\hat{A}_{t} = d\hat{A}_{t} dt = (dt)^{2} = 0
	,\\
	d\hat{A}_{t}  d\hat{A}_{t}^{\ast} = dt
	\end{array} \right.
	.
	\label{eq_quantum_ito_rule}
	\end{align}	
	Then the time evolution of the $\hat{X}_{t} = \hat{U}_{t}^{\ast} \hat{X} \hat{U}_{t}$ 
	is given by following quantum stochastic differential equation
	\begin{align}
	d\hat{X}_{t} 
	=&
	\mathrm{i} [\hat{H}_{t} , \hat{X}_{t}]_{-} dt
	\nonumber \\ & \ 
	+
	\frac{1}{2} \left(
	\hat{L}^{\ast}_{t} [\hat{X}_{t} , \hat{L}_{t}]_{-} 
	+
	[\hat{L}^{\ast}_{t} , \hat{X}_{t} ]_{-} \hat{L}_{t} 
	\right) dt
	\nonumber \\ 
	&\quad  +
	[\hat{L}_{t}^{\ast} , \hat{X}_{t}]_{-} d\hat{A}_{t}
	+
	[\hat{X}_{t} , \hat{L}_{t}]_{-} d\hat{A}_{t}^{\ast}
	\label{eq_physical_quantity_evolution}
	\end{align}
	where $\hat{H}_{t} = \hat{U}_{t}^{\ast} \hat{H} \hat{U}_{t}$ and 	$\hat{L}_{t} = \hat{U}_{t}^{\ast} \hat{L} \hat{U}_{t}$.  
%	, and $[\hat{A},\hat{B}]_{-} := \hat{A} \hat{B} - \hat{B}\hat{A}$.  
	
	We consider the homodyne detection as a detection of the probe system \cite{Gardiner_quantum_noise,walls2008quantum}.  
	The measurement outcome is represented by 
	$\hat{Y}_{t} := \hat{U}_{t}^{\ast} ( \hat{A}_{t} + \hat{A}_{t}^{\ast} ) \hat{U}_{t}$ and its increment is  
	\begin{align}
	d\hat{Y}_{t} 
	=&
	\left(
	\hat{L}_{t} + \hat{L}_{t}^{\ast}
	\right) dt
	+
	 d\hat{A}_{t}
	+
	 d\hat{A}_{t}^{\ast}
	 .
	 \label{eq_observed_signal}
	\end{align}
	
	From the definitions of the unitary operator and the observed process, following equations hold;   
	\begin{align}
	\hat{X}_{t} \hat{Y}_{s} =& \hat{Y}_{s} \hat{X}_{t} , \quad \forall t \geq s \geq 0,
	\\
	\hat{Y}_{t} \hat{Y}_{s} =& \hat{Y}_{s} \hat{Y}_{t} , \quad \forall t , s \geq 0
	\label{eq_classical_process_condition}
	. 
	\end{align}
	We denote the von Neumann subalgebra generated by $\{ \hat{Y}_{s} \} _{s=0}^{t}$ by $\mathcal{Y}_{t}$, which corresponds to the $\sigma$--field generated by measurement record up to time $t$.  
	Clearly, $\mathcal{Y}_{t}$ is a commutative von Neumann subalgebra and 
	$\hat{X}_{t} \in \mathcal{Y}_{t}^{\prime}$ for $t\geq 0$.

%%%%	
	\subsection{Quantum filter}
%%%%	
	
	Let $\hat{\pi}_{t} (\hat{X} ) := \mathbb{P}_{\hat{\rho}} [ \hat{X}_{t} \ | \ \mathcal{Y}_{t} ] $ be a quantum conditional expectation up to $t\geq 0$.  
	The quantum filtering equation is given by following equation \cite{Bouten_quant_filtering}; 
	\begin{align}
	d\hat{\pi}_{t}(\hat{X}) 
	=&
	\hat{\pi}_{t} \Big{(}
	\mathrm{i} [\hat{H} , \hat{X} ]_{-}
	\Big{)}
	 dt
	\nonumber \\
	& \ 
	+
	\frac{1}{2} \hat{\pi}_{t} 
	\Big{(}
	\hat{L}^{\ast} [\hat{X} , \hat{L}]_{-} 
	+
	[\hat{L}^{\ast} , \hat{X} ] _{-} \hat{L} 
	\Big{)} dt
	\nonumber \\
	&
	 \quad +
	\hat{\pi}_{t}\Big{(}
	(\hat{L} - \hat{\pi}_{t}(\hat{L}) )^{\ast} \hat{X} 
	+ 
	\hat{X} (\hat{L} - \hat{\pi}_{t}(\hat{L}) )
	\Big{)}
	\nonumber \\
	&\quad \quad \times 
	( d\hat{Y}_{t}  - \hat{\pi}_{t}(\hat{L} + \hat{L}^{\ast})dt )
	.
	\label{eq_quantum_filtering}
	\end{align}
	Remember that $\mathcal{Y}_{t}$ is identified to a set of classical random variables of the classical probability space $(\Omega , \mathcal{F} , {\rm P})$, 
	there exists $\hat{\rho}_{t}(\omega )  \in \mathcal{S}(\mathcal{H}_{S})$ for all $\omega \in \Omega$  satisfies 
	\begin{align*}
	\hat{\pi}_{t}(\hat{X} ) (\omega ) = \mathrm{Tr}[ \hat{\rho}_{t}(\omega ) \hat{X} ],\quad \forall \hat{X} \in \mathcal{L}(\mathcal{H}_{S}), \ \forall \omega \in \Omega
	. 
	\end{align*}
	By the cyclic property of the trace, the stochastic differential equation of $\hat{\rho}_{t}$, so-called 
	the stochastic master equation or quantum trajectory equation, is given by 
	\begin{align}
	d\hat{\rho} _{t}
	=&
	-\mathrm{i} \left[ \hat{H},  \hat{\rho}_{t} \right] _{-} dt
	\nonumber \\ 
	& \quad  
	+
	\left(
	\hat{L} \hat{\rho}_{t} \hat{L}^{\ast}
	-
	\frac{1}{2} \hat{L}^{\ast} \hat{L} \hat{\rho} _{t}
	-
	\frac{1}{2} \hat{\rho} _{t} \hat{L}^{\ast} \hat{L} 
	\right)
	dt
	\nonumber \\
	&
	+
	\left(
	\hat{L} \hat{\rho}_{t} + \hat{\rho}_{t} \hat{L}^{\ast}
	-
	\mathrm{Tr} \left[ (\hat{L} + \hat{L}^{\ast}) \hat{\rho}_{t} \right] \hat{\rho}_{t}
	\right)
	\nonumber \\
	&\quad \quad \times 
	\left( dy_{t} - \mathrm{Tr} \left[ (\hat{L} + \hat{L}^{\ast})\hat{\rho}_{t} \right] dt \right)
	.
	\label{eq_stochastic_master_equation}
	\end{align}

\section{MAIN RESULT: QUAMTUM SMOOTHING}
	
%%%%	
	\subsection{The proposal quantum smoother}
%%%%	
	
	We deal with a fixed point smoothing problem.  
	One of the simplest quantum smoothing setup is the target quantum physical quantity does not evolve under the unitary operator.  
	That implies $[\hat{U}_{t} , \hat{X}]_{-}=0$ for all $t\geq 0$ and this is called {\it the Braginsky's quantum nondemolition detection condition} \cite{Braginsky_QND}.  
	In this paper, we consider more general setup.  
	Let us derive the recursive expression of the quantum least mean square estimation.  
	As we mentioned in Proposition \ref{prop_least_mean_squareness}, 
	there are two measures to describe least mean square errors.  
	However, the quantum least mean square estimation is composed of the symmetric and skew parts of the quantum least mean square estimate.  
	Therefore we focus on deriving the recursive equations of the symmetric quantum least mean square estimation and the mean non-commutativity.  
	
	Consider to estimate $\hat{X}_{\tau}$, which is the solution of Eq. \eqref{eq_physical_quantity_evolution} at a fixed $\tau \geq 0$, from measurement records ${\mathcal{Y}}_{t}$ up to $t \geq \tau$.  
	Remember that any element of ${\mathcal{Y}}_{t}$ is regarded as a classical random variable, 
	we can apply the martingale method \cite{krishnan1984nonlinear} for deriving the dynamical estimator.  
%	As the rigorous mathematical derivation and jargons make us confuse, 
	We give a sketch to derive the dynamical estimator.  
	We denote $\mathbb{Q}_{t}^{\pm} \left( \hat{X} \right)  := \mathbb{Q}_{\hat{\rho}}^{\pm} \left[ \hat{X} \ | \ \mathcal{Y}_{t} \right]$ for $\hat{X} \in \mathcal{L}( \mathcal{H} )$.  
	Since the physical quantity does not evolve after $\tau $ and ${\mathbb{Q}}^{\pm}_{t} \left( \hat{X}_{\tau} \right) \in {\mathcal{Y}}_{t}$ for all $t\geq \tau$, 
	the process $\{ {\mathbb{Q}}_{t}^{\pm} \left( \hat{X}_{\tau} \right) \} _{t\geq \tau }$ is martingale.  
	Any increment of any martingale process can be represented by multiplication between the increment of the innovation and a uniquely determined coefficient derived from measurement records ({\it the Fujisaki--Kallianpur--Kunita's theorem}).  
	\begin{align}
	d {\mathbb{Q}}_{t}^{\pm} \left( \hat{X}_{\tau} \right) 
	= &
	{\mathbb{Q}}_{t+dt}^{\pm} \left( \hat{X}_{\tau} \right) - {\mathbb{Q}}_{t}^{\pm}\left ( \hat{X}_{\tau} \right)
	\nonumber \\
	=&
	\hat{\Gamma} _{t}^{\pm} \left( d\hat{Y}_{t} - \hat{\pi}_{t} \left( \hat{L} + \hat{L}^{\ast} \right) dt\right)
	\label{eq_quantum_fujisaki_kallianpur_kunita}
	\end{align}
	Note that $\mathbb{Q}_{\tau}^{+}\left( \hat{X}_{\tau} \right) = \hat{\pi}_{\tau} \left( \hat{X} \right) $ and 
	$\mathbb{Q}_{\tau}^{-} \left( \hat{X}_{\tau} \right) = 0 $.  
	Then it remains to derive the coefficient $\hat{\Gamma}_{t}^{\pm} \in {\mathcal{Y}}_{t}$.

	%%%%%%%%%%%%%%%%%%%%%%%%%%%%%%%%
	\begin{thm}
	\label{thm_quantum_smoothing}
	
	Let $\mathbb{Q}_{\tau}^{+} \left( \hat{X}_{\tau} \right) = \hat{\pi}_{\tau} \left( \hat{X} \right)$ 
	and $\mathbb{Q}_{\tau}^{-} \left( \hat{X}_{\tau} \right) = 0 $.  
	Then the recursive estimators are represented by following equations;  
	\begin{align}
	d\mathbb{Q}_{t}^{\pm} \left( \hat{X}_{\tau} \right) =& 
	\frac{1}{2}
		\Big{\{}
	\mathbb{Q} _{t} ^{+} 
	\left( \left[ \left( \hat{L}_{t} + \hat{L}_{t}^{\ast} \right) , \hat{X}_{\tau} \right] _{\pm} \right) 
	\nonumber \\ & \quad
	+
	\mathbb{Q} _{t} ^{-} \left(
	 \left[ 
	  \left( \hat{L}_{t} + \hat{L}_{t}^{\ast} \right) , \hat{X}_{\tau}  
	 \right] _{\mp} \right)
	\nonumber \\ & \quad \quad 
	-2
	{\mathbb{Q}}_{t}^{\pm} \left( \hat{X}_{\tau} \right)  \hat{\pi}_{t} \left( \hat{L} + \hat{L}^{\ast} \right)
	\Big{\}}
	\nonumber  \\ & \quad \times 
	  \left( d \hat{Y}_{t} - \hat{\pi}_{t} \left( \hat{L} + \hat{L}^{\ast} \right) dt \right)
	  , \quad \forall t\geq \tau 
	  \label{eq_quantum_smoothing_equation}
	\end{align}

	\end{thm}
	%%%%%%%%%%%%%%%%%%%%%%%%%%%%%%%%
	
	From Proposition \ref{prop_least_mean_squareness}, the sum of solutions $\mathbb{Q}_{t} \left( \hat{X}_{\tau} \right) = \mathbb{Q}_{t}^{+} \left( \hat{X}_{\tau} \right)  + \mathbb{Q}_{t}^{-} \left( \hat{X}_{\tau} \right) $ is the optimal estimate of $\hat{X}_{\tau}$ in the sense of semi-norm defined by $\langle \bullet , \bullet \rangle _{\hat{\rho}}$.  
	We call $d \mathbb{Q}_{t} \left( \hat{X}_{\tau} \right) $ {\it the quantum smoother}.  
	The symmetric part of the quantum smoother is also the optimal estimate of $\hat{X}_{\tau}$ in the sense of semi-norm defined by $\langle \langle \bullet , \bullet \rangle \rangle _{\hat{\rho}}$.  
	We call $d\mathbb{Q}_{t}^{+}(\hat{X}_{\tau})$ {\it the symmetric quantum smoother}.  
	We obtain two different optimal smoothers.  
	
	\begin{proof}
	To derive $\hat{\Gamma}_{t}^{\pm}$, we refer a method in \cite[Sec. 5]{Bouten_discrete_quant_filtering}.  
	From Definition \ref{def_least_square_approximation}, 
	${\mathbb{P}}_{\hat{\rho}} \left[  \hat{Z} \hat{Y}_{t} \hat{X}_{\tau} \pm \hat{X}_{\tau} \hat{Y}_{t} \hat{Z}   \right] 
	= 
	2 {\mathbb{P}}_{\hat{\rho}} \left[ {\mathbb{Q}}_{t}^{\pm} \left( \hat{X}_{\tau} \right) \hat{Y}_{t} \hat{Z} \right]$
	holds for every $t\geq \tau$ and any $\hat{Z} \in {\mathcal{Y}}_{t}$ because $[\hat{Y}_{t}, \hat{Z} ]_{-} =0 $ and $\hat{Y}_{t}\hat{Z} \in \mathcal{Y}_{t}$.  
	This implies that 
	\begin{align}
	 {\mathbb{P}}_{\hat{\rho}} 
	 \left[ 
	 d \left( \hat{Z} \hat{Y}_{t} \hat{X}_{\tau} \pm \hat{X}_{\tau} \hat{Y}_{t} \hat{Z} \right) 
	 \right] 
%	\nonumber \\ 
	=&  
	2 {\mathbb{P}}_{\hat{\rho}} \left[ 
	d \left( {\mathbb{Q}}^{\pm}_{t} (\hat{X}_{\tau})  \hat{Y}_{t} \right)  \hat{Z} 
	\right]
	,
	\label{eq_siceannual2014_derivation_1}
	\end{align}
	$\forall \hat{Z} \in \mathcal{Y}_{t}$ holds.  
	Thanks to the quantum Ito's formula \cite{Hudson_Parthasarathy}, 
	the right-hand side of Eq. \eqref{eq_siceannual2014_derivation_1} is 
	\begin{align*}
	& {\mathbb{P}}_{\hat{\rho}} \left[ 
	d \left( {\mathbb{Q}}_{t}^{\pm} \left( \hat{X}_{\tau} \right)  \hat{Y}_{t} \right)  \hat{Z} \right]
	\\
	=&
	{\mathbb{P}}_{\hat{\rho}} \left[ 
	{\mathbb{Q}}_{t}^{\pm} \left( \hat{X}_{\tau} \right)  \left( \hat{L} _{t} + \hat{L}_{t}^{\ast} \right)
	\hat{Z}
	\right] 
	dt
	+
	{\mathbb{P}}_{\hat{\rho}}
	\left[ 
	\hat{\Gamma}_{t}^{\pm}
	\hat{Z}
	\right] 
	dt
	\\
	=&
	{\mathbb{P}}_{\hat{\rho}} \left[ 
	{\mathbb{Q}}_{t}^{\pm} (\hat{X}_{\tau})  \hat{\pi}_{t} \left( \hat{L} + \hat{L}^{\ast} \right)
	\hat{Z}
	\right] 
	dt
	+
	{\mathbb{P}}_{\hat{\rho}}
	\left[ 
	\hat{\Gamma}_{t} ^{\pm}
	\hat{Z}
	\right] 
	dt
	.
	\end{align*}
	Then the left-hand side of Eq. \eqref{eq_siceannual2014_derivation_1} is 
	\begin{align*}
	& {\mathbb{P}}_{\hat{\rho}} \left[ 
	\hat{Z} d \hat{Y}_{t} \hat{X}_{\tau}  \pm \hat{X}_{\tau} d\hat{Y}_{t} \hat{Z} 
	\right] 
	\\
	=&
	{\mathbb{P}}_{\hat{\rho}} \left[ 
	\hat{Z} \left( \hat{L}_{t} + \hat{L}_{t}^{\ast} \right) \hat{X}_{\tau}   
      \pm 
	\hat{X}_{\tau} \left( \hat{L}_{t} + \hat{L}_{t}^{\ast} \right) \hat{Z} 
	\right] 
	dt.  
	\end{align*}
%	The least mean square estimation is not applied directly, though, 
	Note that the following decomposition gives us the estimations.  
	\begin{align*}
	&{\mathbb{P}}_{\hat{\rho}} \left[ 
	\hat{Z} \left( \hat{L}_{t} + \hat{L}_{t}^{\ast} \right) \hat{X}_{\tau}   
      \pm 
	\hat{X}_{\tau} \left( \hat{L}_{t} + \hat{L}_{t}^{\ast} \right) \hat{Z} 
	\right] 
	\nonumber \\
	=&
	\frac{1}{2}
	{\mathbb{P}}_{\hat{\rho}} \left[ 
	\hat{Z} 
	\left[ \left( \hat{L}_{t} + \hat{L}_{t}^{\ast} \right) , \hat{X}_{\tau} \right]_{\pm} 
	+ 
	\left[  \left( \hat{L}_{t} + \hat{L}_{t}^{\ast} \right) , \hat{X}_{\tau} \right]_{\pm}
	\hat{Z} 
      \right] 
	\nonumber \\
	& 
	+ 
	\frac{1}{2}
	{\mathbb{P}}_{\hat{\rho}} \left[
	\hat{Z} 
	\left[ \left( \hat{L}_{t} + \hat{L}_{t}^{\ast} \right) , \hat{X}_{\tau} \right] _{\mp} 
	- 
	\left[  \left( \hat{L}_{t} + \hat{L}_{t}^{\ast} \right) , \hat{X}_{\tau} \right]_{\mp}
      \hat{Z} 
      \right] 
      \\
      =& 
      {\mathbb{P}}_{\hat{\rho}} \left[ 
	\mathbb{Q} _{t} ^{+} \left( \left[ \hat{X}_{\tau} , \left( \hat{L}_{t} + \hat{L}_{t}^{\ast} \right) \right] _{\pm} \right) 
	\hat{Z} 
      \right] 
      \\ & \quad 
      +
      {\mathbb{P}}_{\hat{\rho}} \left[ 
	\mathbb{Q} _{t} ^{-} \left( \left[ \hat{X}_{\tau} , \left( \hat{L}_{t} + \hat{L}_{t}^{\ast} \right) \right] _{\mp} \right) 
	\hat{Z} 
      \right] 
	.  
	\end{align*}
	Then we obtain 
	\begin{align*}
	\hat{\Gamma} _{t}^{\pm} = & 
	\frac{1}{2}\mathbb{Q} _{t} ^{+} 
	\left( \left[  \left( \hat{L}_{t} + \hat{L}_{t}^{\ast} \right) \hat{X}_{\tau} \right] _{\pm} \right) 
	\\ & \quad 
	-
	{\mathbb{Q}}_{t}^{\pm} (\hat{X}_{\tau})  \hat{\pi}_{t} \left( \hat{L} + \hat{L}^{\ast} \right)
	\\ & \quad \quad 
	+
	\frac{1}{2}\mathbb{Q} _{t} ^{-} 
	\left( \left[ \left( \hat{L}_{t} + \hat{L}_{t}^{\ast} \right) , \hat{X}_{\tau} \right] _{\mp} \right) 
	,
	\end{align*}
	and this gives Eq. \eqref{eq_quantum_smoothing_equation}.  
	\end{proof}

	From \eqref{eq_quantum_smoothing_equation}, it is necessary to compute the skew part even if we want to compute the symmetric quantum smoother.  
%	For implementation of the quantum smoother 
%	\eqref{eq_quantum_smoothing_equation} for general quantum systems is not easy 
%	because 
	To compute the \eqref{eq_quantum_smoothing_equation}, we have to know 
	$\mathbb{Q} _{t} ^{\pm} \left( 
	\left[  \left( \hat{L}_{t} + \hat{L}_{t}^{\ast} \right) , \hat{X}_{\tau} \right]_{\pm} \right) $. 
	The time evolution equations of these operators can be derived in similar way, so we omit the derivation.  
	Note that if we obtain the equations of these operators, it is necessary to compute other equations to compute estimates of other operators in general.

	\subsection{Density operator description}
			
	Unlike the quantum filtering theory, it is difficult to find equations corresponding to the stochastic master equation.  
	If we consider the Bragynski's quantum nondemolition detection condition \cite{Braginsky_QND}, 
	the smoothing equations can be represented by the equations corresponding to stochastic master equation.  
%	we can derive the equations corresponding to stochastic master equation.  
	Suppose that the operator $\hat{L} \in \mathcal{L}(\mathcal{H}_{S})$ is normal 
	and satisfies $[\hat{L} , \hat{U}_{t} ] _{-} = 0 $, for all $t\geq 0$.   
	In this case, $\hat{\pi}_{t}(\hat{L} ) = \mathbb{Q}_{t}^{+}(\hat{L})$ and there exist 
	$\hat{\rho}_{0|t}^{+} (\omega )$ and $\hat{\rho}_{0|t}^{-} (\omega )$ satisfies 
	$\mathbb{Q}_{t}^{\pm} \left( \hat{Z} \right)  (\omega ) = \mathrm{Tr}[ \hat{\rho}_{0|t}^{\pm}(\omega ) \hat{Z}  ] $ 
	for all $\omega \in \Omega $ and $\hat{Z} \in \mathcal{L}(\mathcal{H}_{S})$.  
	Then, for $\hat{X}_{0} = \hat{X}$, we can compute the following equations in stead of the quantum smoother;  
	\begin{align}
	d\hat{\rho}_{0|t}^{\pm} =&
	\frac{1}{2}
	\Big{\{}
	\left[  \hat{\rho}_{0 |t}^{+} , \left( \hat{L} + \hat{L}^{\ast} \right) \right] _{\pm}
	+
	\left[ \hat{\rho}_{0 |t}^{-} , \left( \hat{L} + \hat{L}^{\ast} \right) \right] _{\mp}
	\nonumber \\ & \quad \quad 
	-2 \mathrm{Tr}[ \hat{\rho}_{t} ^{+} (\hat{L} + \hat{L}^{\ast} ) ] \hat{\rho}_{0 | t}^{\pm}
	\Big{\}}
	\nonumber \\ & \hspace{2cm} 
	\times 
	\left( dy_{t} 
	-
	\mathrm{Tr}[ \hat{\rho}_{t} (\hat{L} + \hat{L}^{\ast} ) ] dt
	\right)
	\end{align}
	where $\hat{\rho}_{t}$ is the solution of the stochastic master equation \eqref{eq_stochastic_master_equation}.  

\section{CONCLUSIONS}

	We developed a new quantum smoothing theory for fixed smoothing problems.  
	It is characterized by two kinds of orthogonality based on two pre--inner products.  
	As a result, we obtain two smoothers; 
	one is the complex valued estimator and the other is real valued estimator for quantum physical quantities.  
	To develop numerical methods for implementation of the smoothers is necessary in the future.

%%%%%%%%%%%%%%%%%%%%%%%%%%%%%%%%%%%%%%%%%%%%%%%%%%%%%%%%%%%%%%%%%%%%%%%%%%%%%%%%

%%%%%%%%%%%%%%%%%%%%%%%%%%%%%%%%%%%%%%%%%%%%%%%%%%%%%%%%%%%%%%%%%%%%%%%%%%%%%%%%

%	
%	\section*{APPENDIX}
%	\subsection*{Proof of Proposition \ref{prop_least_mean_squareness}}
%	\label{proof_prop_least_mean_squareness}
%	

%%%%%%%%%%%%%%%%%%%%%%%%%%%%%%%%%%%%%%%%%%%%%%%%%%%%%%%%%%%%%%%%%%%%%%%%%%%%%%%%

\addtolength{\textheight}{-12cm}   % This command serves to balance the column lengths

%\section*{ACKNOWLEDGMENT}
%
%The author is 

%%%%%%%%%%%%%%%%%%%%%%%%%%%%%%%%%%%%%%%%%%%%%%%%%%%%%%%%%%%%%%%%%%%%%%%%%%%%%%%%

%\bibliographystyle{unsrt}        % Include this if you use bibtex 
%\bibliography{bib/bib_information,bib/bib_prob,bib/bib_quantum,bib/bib_phd_thesis,bib/bib_math,bib/bib_control,bib/bib_physics,bib/bib_filtering,bib/bib_quant_contr,bib/bib_engineering,bib/bib_quantum_correlation}
           % and a bib file to produce the 

%\begin{thebibliography}{99}
%
%\bibitem{c1} test
%
%
%
%
%\end{thebibliography}

\end{document}